%% file: rrma.tex
\documentclass{llncs}
\usepackage[numbers]{natbib}

\usepackage{amsmath,enumerate,fancyhdr,amssymb,amsfonts,hyperref}
\usepackage[ruled]{algorithm2e}
\usepackage{graphicx}           
\usepackage{myDefs,cleveref}
\newcommand*\samethanks[1][\value{footnote}]{\footnotemark[#1]}

\newcommand{\porm}{\mbox{\sc PoRM}}
\newcommand{\ppi}{\mbox{\sc ppi}}
\newcommand{\tp}[2]{\left((#1),#2\right)}
\newcommand{\ck}{\lceil k/2 \rceil}
\newcommand{\fk}{\lfloor k/2 \rfloor}
\newcommand{\rev}{\mbox{\sc revenue}}
\newcommand{\wf}{\mbox{\sc welfare}}
\newcommand{\mprg}{\mbox{\sc mprg}}

\crefname{equation}{Equation}{Equations}
\crefname{theorem}{Theorem}{Theorems}
\crefname{definition}{Definition}{Definitions}
\crefname{figure}{Fig}{Figs}
\crefname{section}{Section}{Sections}
\crefname{lemma}{Lemma}{Lemmas}
\crefname{proposition}{Proposition}{Propositions}
\crefname{appendix}{Appendix}{Appendixes}
\crefname{nclaim}{Claim}{Claims}
\crefname{claim}{Claim}{Claims}
\crefname{observation}{Observation}{Observations}
\crefname{enumi}{Condition}{Conditions}

\begin{document}
\mainmatter  

\title{Randomized Revenue Monotone Mechanisms for Online Advertising}

\author{Gagan Goel \inst{1}
\and MohammadTaghi Hajiaghayi \inst{2} \thanks{supported in part by NSF CAREER award 1053605, ONR YIP award N000141110662, DARPA/AFRL award FA8650-11-1-7162.}\and Mohammad Reza Khani\inst{2} \samethanks}

\institute{
Google Inc., New York, NY. \protect\url{gagangoel@google.com}
\and University of Maryland, College Park, MD. \protect\url{hajiagha@cs.umd.edu}.
\and University of Maryland, College Park, MD. \protect\url{khani@cs.umd.edu}.
}

\maketitle

\input{abstract.tex}

%
%
%

%

\input{introduction.tex}
\input{related.tex}
\input{results.tex}
\input{prelim.tex}
\input{rcaii.tex}
\bibliographystyle{plainnat}
\bibliography{refs}

\newpage
\appendix
\input{multigroup.tex}

\end{document}

%% file: abstract.tex
\begin{abstract}
Online advertising is the main source of revenue for many Internet firms. A central component of online advertising is the underlying mechanism that selects and prices the winning ads for a given ad slot. In this paper we study designing a mechanism for the Combinatorial Auction with Identical Items (CAII) in which we are interested in selling $k$ identical items to a group of bidders each demanding a certain number of items between $1$ and $k$. CAII generalizes important online advertising scenarios such as image-text and video-pod auctions \cite{GK14}. In image-text auction we want to fill an advertising slot on a publisher's web page with either $k$ text-ads or a single image-ad and in video-pod auction we want to fill an advertising break of $k$ seconds with video-ads of possibly different durations.

Our goal is to design truthful mechanisms that satisfy Revenue Monotonicity (RM). RM is a natural constraint which states that the revenue of a mechanism should not decrease if the number of participants increases or if a participant increases her bid.

In a recent work \citet{GK14}, it was argued that RM is a desired property to have for smooth functioning of a firm. Since popular mechanisms like VCG are not revenue-monotone, they introduced the notion of Price of Revenue Monotonicity (\porm{}) to capture the loss in social welfare of a revenue-monotone mechanism. \citet{GK14} showed that no deterministic RM mechanism can attain \porm{} of less than $\ln(k)$ for CAII, \ie, no deterministic mechanism can attain more than $\frac{1}{\ln(k)}$ fraction of the maximum social welfare. \citet{GK14} also design a mechanism with \porm{} of $O(\ln^2(k))$ for CAII.

In this paper, we seek to overcome the impossibility result of \citet{GK14} for deterministic mechanisms by using the power of randomization. We show that by using randomization, one can attain a constant \porm{}. In particular, we design a randomized RM mechanism with \porm{} of $3$ for CAII.



Finally we study Multi-group Combinatorial Auction with Identical Items (MCAII) which is an important generalization of CAII. In MCAII the bidders are partitioned into multiple groups and the set of winners should be from a single group. The motivation for MCAII is from online advertising scenarios where, for instance, the set of selected ads may be required to have the same format. We give a randomized mechanism which satisfies RM and IC and has \porm{} of $O(\ln k)$. This is in contrast to $\log^2(k)$ deterministic mechanism that follows from \cite{GK14}. 
\end{abstract}

%% file: introduction.tex
\section{Introduction}
Many Internet firms including search engines, social networks, and online publishers rely on online advertising revenue for their business; thus, making online advertising an essential part of the Internet. Online advertising consists of showing a few ads to a user when she accesses a web-page from a publisher's domain. The advertising can happen in different formats such as text-ads, image-ads, video-ads, or a hybrid of them.

A key component in online advertising is a mechanism which selects and prices the set of winning ads. In this paper we study the design of mechanisms for Combinatorial Auction with Identical Items (CAII). In CAII we want to sell $k$ identical items to a group of bidders; each demand a number of items from $\{1, \ldots, k\}$ and has a single-parameter valuation for obtaining them. Although CAII is a well-motivated model on its own, we note that a few important advertising scenarios such as image-text and video-pod auctions can be modeled by CAII. In image-text auction we want to fill an advertising box on a publisher's web-page with either one image-ad or $k$ text-ads. We note that a large portion of Google AdSense's revenue is from this auction. Image-text auction is a special case of CAII where participants either demand only one item (text-ads) or all $k$ items (image-ads). In video-pod auction there is an advertising break of $k$ seconds which should be filled with video-ads each with certain duration and valuation.

When designing a mechanism, typically one focusses on attaining incentive-compatibility, and maximizing social welfare and/or revenue. In a recent work, \citet{GK14} argue that the mechanisms for online advertising should satisfy an additional property of revenue-monotonicity.  Revenue-monotonicity is a natural property which states that the revenue of a mechanism should not decrease as the number of bidders increase or if the bidders increase their bids. The motivation is that any online firm typically has a large sales team to attract more bidders on their inventory or they invest in new technologies to make bids more attractive. The typical reasoning is that more bidders (or higher bids) lead to more competition which should lead to higher prices. However, lack of revenue-monotonicity of a mechanism is conflicting with this intuitive and natural reasoning process, and can create significant confusion from a strategic decision-making point of view.

Even though Revenue Monotonicity (RM) seems very natural, we note that majority of the well-known mechanisms do not satisfy this property \cite{RCL09,RCL11,GK14}. For example the famous Vickrey-Clarke-Groves (VCG) mechanism fails to satisfy RM as adding one more bidder might decreases the revenue to zero. To see this, consider two identical items to be sold to two bidders. One wants one item with a bid 2, and the other one wants both items with a bid 2. In this case the revenue of VCG mechanism is 2 (for a proof, see for instance \cite[Chapter 9]{NRTV07}). Now suppose we add one more bidder who wants one item with a bid of 2. In this case the revenue of VCG goes down to 0!

It is known that if we require mechanisms to satisfy both RM and IC, not only the mechanism cannot get the maximum social welfare but it can also not achieve Pareto-optimality in social welfare \cite{RCL11}. In light of this, \cite{GK14} introduced the notion of {\em Price of Revenue Monotonicity} (\porm{}) to capture the loss in social welfare for RM mechanisms. Here a mechanism has \porm{} of $\alpha$ if its social welfare is at least $\frac{1}{\alpha}$ fraction of the maximum social welfare in any type profile of participants.   It is shown that, under a mild condition, the \porm{} of any deterministic mechanism for the CAII problem is at least $\ln(k)$, \ie, no deterministic mechanism can obtain more than $\frac{1}{\ln(k)}$ fraction of the maximum social welfare \cite{GK14}. In fact this impossibility result holds even for the case when participants demand either all the items or only one item. On the positive side, \cite{GK14} give a deterministic mechanism with \porm{}  of $O(\ln^2(k))$ for CAII. We note that satisfying RM is hard especially since it is an {\em across instance} constraint.

This work is motivated by the desire to design better mechanisms for CAII. However, the above impossibility result of \cite{GK14} is a bottleneck towards this goal. To overcome this, in this paper, we resort to randomized mechanisms. We say a randomized mechanism satisfies RM if it satisfies RM in expectation\footnote{Since in a  typical online advertising setting, there is a large number of auctions being run everyday, we get sharp concentration bounds.}. Similarly, a randomized mechanism has \porm{} of $\alpha$ if its expected social welfare is not less than $\frac{1}{\alpha}$ fraction of the maximum social welfare.  We significantly improve the performance by designing a randomized mechanism with a constant \porm{}. In particular, our randomized mechanism achieves a \porm{} of $3$.

Finally, we study Multi-group Combinatorial Auction with Identical Items (MCAII) that generalizes CAII.
In MCAII bidders are partitioned into multiple groups and the set of winners has to be only from one group. The motivation is that the publisher sometimes require the ads to be of same format or size for a given ad slot. We design a randomized mechanism for MCAII that satisfies IC and RM with \porm{} $O(\log k)$. An easy corollary of \cite{GK14} gives a deterministic mechanism with a \porm{} $O(\log^2 k)$. 

%% file: related.tex
\section{Related Works}
\citet{GK14} show that RM is a desirable property for web-centeric companies and consider designing mechanisms which satisfy both RM and IC. They introduced the notion of  \porm{} and study CAII and a special case of it - namely, image-text auction. They \cite{GK14} give a deterministic mechanism with \porm{} of $\ln(k)$ and prove that no mechanism which satisfies RM and IC can obtain \porm{} of better than $\ln(k)$ under the following two mild conditions. The first condition is anonymity which states that the outcome shouldn't depend on the identities of the bidders but their type profile. The second condition is independence of irrelevant alternatives which states that decreasing the bid of any losing participant should not hurt a winning participant. \citet{GK14} also give a deterministic mechanism for CAII with \porm{} of $O(\ln^2(k))$ that satisfies IC and RM.

\citet{RCL11} show that for combinatorial auctions, no deterministic mechanism that satisfies RM and IC can get weak maximality. A mechanism is Weakly Maximal (WM) if it chooses an allocation which cannot be augmented to make a losing participant a winner without hurting a winning participant. \citet{RCL09} study randomized mechanisms for combinatorial auctions which satisfy RM and IC. Note that a simple mechanism which chooses a maximal allocation uniformly at random ignoring the valuations of bidders satisfies RM, IC, and WM. \citet{RCL09} add another constraint that a mechanism has to also satisfy Consumer Sovereignty (CS) which means that if a bidder increases her bid high enough, she can win her desired items. Now a new issue is that there is no randomized mechanism which satisfies RM, IC, WM, and CS \cite{RCL09}. In order to avoid this issue they relax CS constraint as follows. For each participant $i$ there has to be $\lambda$ different valuations $v_1 > v_2 > ... > v_{\lambda}$ such that for $j \in \{1, \ldots, \lambda\}$, we have $w_i(v_j) > w_i(v_{j+1}) + \sigma$ where $w_i$ is the probability of winning for participant $i$ and $\sigma > 0$. Roughly speaking relaxed CS constraint means that if participant $i$ increases her bid from zero to infinity she sees at least $\lambda$ jumps of length $\sigma$ in her winning probability. The idea of their mechanism is that for each participant $i$ they find $\lambda$ constant values $c_{i,1}> c_{i,2}> ...> c_{i, \lambda}$ such that regardless of valuations of the other bidders; if the bid of bidder $i$ is between $c_{i,j}$ and $c_{i, j+1}$ then her winning probability is at least $j * \sigma$.
In order to find the constants for each participant they solve a LP whose constraints force RM, IC, Relaxed CS, and WM. 
As you may notice although this mechanism achieves WM, RM and relaxed CM, but can do very poorly in terms of \porm{}. For example suppose you have $n$ participants and each of them wants all items. The valuation of each participant $i$ is bigger than its highest constant $c_{i,1}$. In this case all the participants can win with probability at most $1/n$. Now suppose that the valuation of one of the participants is infinity. She still wins with probability $1/n$ which shows that the PoRM of their mechanism is at least $n$.

\citet{DRS09} show that VCG is revenue monotone if and only if the feasible subsets of winners form a matroid.
\citet{AM02} show that if valuations of bidders satisfy {\em bidder-submodularity} then VCG satisfies RM. Here valuations satisfy bidders submodularity if and only if for any bidder $i$ and any two sets of bidders $S, S'$ with $S \subseteq S'$ we have $\wf(S \cup \{i\}) - \wf(S) \geq \wf(S' \cup \{i\}) - \wf(S')$, where $\wf(S)$ is the maximum social welfare achievable using only bidders in $S$. Note that we can restrict the set of possible allocations in a way such that bidder-submodularity holds. Then we can use VCG on this restricted set of allocations and hence achieve RM. However we can show that it is not possible to get a mechanism with \porm{} better than $\Omega(k)$ by using the mentioned tool. 

\citet{AM02} design a mechanism which is in the core of the exchange economy for combinatorial auctions. A mechanism is in the core if there is no subset of participants including the seller which can collude and trade among each other such that all of them benefit more than the result of the mechanism.
\citet{DM08} show that a core-selecting mechanism which selects an allocation that minimizes the seller's revenue satisfies RM given bidders follow so called {\em best-response truncation strategy}. Therefore, the mechanism of \cite{AM02} satisfies RM if it selects an allocation that minimizes the seller's revenue and the participants follow best-response strategy, however, this mechanism does not satisfy IC.

Another line of related works is around characterizing incentive compatible mechanisms. The classic result of \citet{R79} tells that affine maximizers are the only social choice functions which can be implemented using mechanisms that satisfy IC when bidders have unrestricted quasi-linear valuations. Subsequent works study some restricted cases, see \eg{} \cite{R87,LMN03,BCL+06,SY05}.

There is also a large body of research around designing mechanisms with good bounds on the revenue.
In the single parameter Bayesian setting \citet{M81} designs a mechanism which achieves the optimal expected revenue. \cite{GHW01,GHK+06} consider optimizing revenue in prior-free settings (see \eg{} \cite{NRTV07} for a survey on this).

%% file: results.tex
\section{Our Results and Overview of Techniques}
To give intuition about our approach, we first start with ideas that will not work but are potentially good candidates. To keep the explanation easier let us focus on deterministic mechanisms. Note that the payment of each participant in a deterministic mechanism which satisfies IC is her critical value, \ie{}, the minimum valuation for which she still remains a winner. Assume that all participants demand only one item. In this case we can simply give all the items to the highest $k$ bidders, which sets the critical value (the payment) of each winner to the valuation of the ($k+1$)th highest bidder. If we add one more participant the valuation of the ($k+1$)th highest bidder increases, therefore, the payment of each winner increases  and hence the mechanism satisfies RM.

Now assume we have two types of bidders: A bidder of type A who demand all $k$ items, and a bidder of type B who demands a single item. This scenario is equivalent to the image-text auction for which there is a lower-bound of $\ln(k)$ for the \porm{} of deterministic mechanisms \cite{GK14}. However using randomization we can simply get a \porm{} of $2$. Flip a coin and with probability half give all items to the highest type A bidder and with probability half give $k$ items to the $k$ highest bidders of type B. Here, the expected social welfare is at least half of the maximum social welfare. Note that when the coin flip selects bidders of type A the auction simply transforms to the second price auction of selling one package of items which has RM. When it selects bidders of type B the auction transforms to the case when all bidders demand one item which we explained earlier and has RM. Therefore, the expected revenue is monotone and hence the mechanism satisfies RM.
Expanding the above idea we can partition the bidders into $\log(k)$ groups such that the bidders of each group $i \in \log(k)$ has demand in $[2^{i}, 2^{i+1})$. Then, we randomly select one group and choose the winners from the selected group. However, this partitioning approach does not lead to a \porm{} better than $\log(k)$.

As a second approach instead of partitioning the bidders and sort them by their valuation, we can sort them according to their Price Per Item (\ppi{}) which is the valuation of a participant divided by the number of items she demands. Now consider a simple greedy algorithm as follows. Start from the top of the sorted list of bidders and at each step do the following. If the number of remaining items is enough to serve the current bidder give the items to the bidder and proceed; otherwise stop. Let us call the bidder at which the greedy algorithm stops the {\em runner-up bidder}. Note that the runner-up bidder  has the largest \ppi{} among the loser bidders and let $p$ be her \ppi{}. If each of the winner participant had \ppi{} less than $p$ then she could not win. Therefore, the critical value of each winner participant is her demand multiplied by $p$. Although value $p$ increases if we add more bidders, the number of items sold might decrease. For example consider the case when the bidder with the highest \ppi{} demands all $k$ items. In this scenario she wins all items and pays $k$ multiplied by the \ppi{} of the runner-up bidder. Now if we add one more bidder whose \ppi{} is more than the highest bidder but demands only one item; the new bidder wins and we sell only one item. This potentially decreases the revenue of the greedy mechanism.

For our mechanism we use a combination of the above ideas and an extra interesting technique. We partition the bidders into two groups: high-demand bidders who demand more than $k/2$ items, and low-demand bidders who demand less than or equal to $k/2$ items. With probability $1/3$ the winner is a high-demand bidder with the largest valuation. Similar to the partitioning approach the critical value of the winner is the second largest valuation of the high-demand bidders which can only increase if we add more bidders. With probability $2/3$ we do the following with the low-demand bidders. First we run the greedy algorithm over the low-demand bidders and find the runner-up bidder. The important observation here is that because there is no high-demand bidder, sum of winners' demands ($A$)  is larger than $k/2$. Therefore we are sure that we sell at least $k/2$ items where the price of each item is the \ppi{} of the runner-up bidder. Now we select each winner of the greedy algorithm with probability $\frac{k/2}{A}$ as the true winner of our mechanism. This random selection makes sure that the expected number of sold items is exactly $k/2$. The exact number $k/2$ is important since the expected revenue of the mechanism is $k/2$ multiplied by the \ppi{} of the runner-up bidder. Therefore as the \ppi{} of the runner-up bidder increases if we add more bidders the expected revenue is monotone. 

Now we explain ideas used to design our mechanism for MCAII. We first note that as a corollary of the result of \cite{GK14}, we get a deterministic mechanism with a \porm{} of $\log^2(n)$. In our mechanism, we assign a value to each group and use it as the criterion in order to select the winner group.  Note that a simple value that can be assigned to each group is the maximum social welfare obtainable by the group. However, this way we cannot guarantee RM. Because suppose participant $i^{(g)}$ of group $G^{(g)}$ increases her bid high enough which guarantees that $G^{(g)}$ wins against all other groups no matter what are the valuations of the other participants of $G^{(g)}$.  Therefore, the critical values of the other members of $G^{(g)}$ decreases as $i^{(g)}$ increases her bid and hence can decrease the revenue of the mechanism.

We refer to our assigned value to each group as the Maximum Possible Revenue of the Group (\mprg{}). As name \mprg{} suggests, it shows the maximum revenue we can obtain from each group without the fear of violating RM.
For each $j \in \{1, \ldots, k\}$ and group $G^{(g)}$, let $u^{(g)}_j$ be the maximum price can be set for a single item so that we can sell at least $j$ items to low-demand bidders of group $G^{(g)}$. More formally, $u^{(g)}_j$ is the maximum value where the sum of demands of low-demand bidders whose \ppi{} is larger than $u^{(g)}_j$ in group $G^{(g)}$ is at least $j$. The \mprg{} of group $G^{(g)}$ is $\max(V^{(g)}, \max_{j \in \{1, \ldots, k/2\}} j\cdot u_j^{(g)})$ where $V^{(g)}$ is the highest valuation of high-demand bidders. Intuitively, \mprg{} either sells items to high-demand bidders and obtains revenue of at most $V^{(g)}$ or sells items to low-demand bidders in which we can sell a number of items between $1$ and $k/2$. We select a group with the highest \mprg{} and choose the winners from this group. We are able to show that we can obtain a revenue of at least the second highest \mprg{}. We prove that our mechanism satisfies RM by showing that the second highest \mprg{} increases if we add more bidders.

We show that the \mprg{} of each group is at least $1/\ln(k)$ fraction of the maximum social welfare obtainable by the group. Therefore, as we select the winning group using the \mprg{}s of groups, the \porm{} of our mechanism is $O(\ln(k))$. We provide evidence that indeed the \mprg{} of each group is the closest value to its social welfare that can be  safely used for selecting the winning group without violating RM. Moreover, any randomization over the groups for selecting the winning one according to \mprg{} cannot improve the \porm{} factor.

%% file: prelim.tex
\section{Preliminary}
\label{sec:prelim}
Let assume we have a set of $n$ bidders $\{1, \ldots, n\}$ and a set of $k$ identical items. Let type profile $\theta$ be a vector containing the type of each bidder $i$ which we show by $\theta_i$. Here $\theta_i$ is pair $(d_i, v_i) \in [k] \times \bbR^+$ where $d_i$ is the number of items she demands and $v_i$ shows her valuation for getting $d_i$ items. Here we assume the demands are publicly known because in our scenario they represent the length of video-ads stored in database while the valuations are private to bidders.

Note that having higher valuation does not necessarily mean that the bidder is more desirable to the seller as she might have a large demand. We define Price Per Item (\ppi{}) of bidder $i$ to be $\frac{v_i}{d_i}$ which we use in our mechanism to compare bidders. 

We show a randomized mechanism ($\M$) by pair $(w, p)$ where $w_i(\theta)$ shows the winning probability of bidder $i$ in type profile $\theta$ and $p_i(\theta)$ is her expected payment.

We use the following Theorem in this paper frequently which is a well-known characteristic of the truthful randomized mechanisms in the single parameter model (see \eg \cite{NRTV07}).
\begin{theorem}
\label{thm:tf}
Randomized mechanism $\M = (w, p)$ is truthful if and only if for any type profile $\theta$ and any bidder $i$ with type $(d_i, v_i)$ the followings hold.
\begin{enumerate}
\item Function $w_i\tp{d_i, v_i}{\theta_{-i}}$ is weakly monotone in $v_i$. \label{cond:1}
\item $p_i(\theta) = v_i \cdot w_i(\theta) - \int_{0}^{v_i} w_i\tp{d_i, t}{\theta_{-i}} dt $ \label{cond:2}
\end{enumerate}
\end{theorem}

%% file: rcaii.tex
\section{Combinatorial Auction with Identical Items}
\label{sec:rcaii}

We build a randomized mechanism ($\M = (w, p)$) satisfying revenue monotonicity and incentive compatibility such that $\porm(\M)$ is equal to $3$.

We call a bidder {\em high-demand} bidder if her demand is greater than $\fk$ otherwise we refer it as {\em low-demand} bidder. Mechanism $\M$ partitions the bidders into two groups of low-demand and high-demand bidders and with probability $1/3$ selects the winning set from the high-demand bidders and with probability $2/3$ from the low-demand bidders. 

We will see that mechanism $\M$ favors high-demand bidders with larger valuations and favors low-demand-bidders with larger $\ppi$s while breaking the ties by the index number of the bidders.
\begin{definition}
\label{def:valuable}
We call low-demand bidder $l_1$ is {\em more valuable} than low-demand bidder $l_2$ and show it by ($l_1 \succ l_2$) if $\ppi_{l_1} > \ppi_{l_2}  \vee (\ppi_{l_1}= \ppi_{l_2} \wedge l_1 < l_2)$. Similarly we call high-demand bidder $h_1$ is {\em more valuable} than high-demand bidder $h_2$ and show it by ($h_1 \succ h_2$) if $v_{h_1} > v_{h_2}  \vee (v_{h_1} = v_{h_2} \wedge h_1 < h_2)$.
\end{definition}

Let's assume that there are $\ell$ low-demand bidders and $h$ high-demand bidders.  By adding some dummy bidders with demand $1$ and valuation zero we assume that the sum of demands of low-demand bidders is always greater than $k$.
Without loss of generality we assume that the first $\ell$ bidders are low-demand bidders and $i \succ i+1$ for any $i \in [\ell - 1]$ (the \ppi{}s of the low-demand bidders decreases by their index) and the remaining $h$ bidders are high-demand-bidders while $i \succ i+1$ for any $i \in \{\ell + 1, \ldots n-1\}$ (the valuations of the high-demand bidders decreases by their index). 

\begin{definition}
\label{def:runner-up}
We call low-demand bidder $r$ the runner-up bidder if $r$ is the smallest value in set $[\ell]$ for which $\sum_{i = 1}^r d_{i} \geq k$. 
\end{definition}

Later we will see that the runner-up bidder is the bidder with the largest \ppi{} and smallest index number who has zero probability of winning. We simply refer to the runner-up bidder as $r$.

We define $A$ to be $\sum_{i = 1}^{r-1} d_i$ which is the sum of demands of low-demand bidders that have \ppi{}s greater than or equal to that of $r$ and have positive probability of winning (see \cref{fig:runner-up}). 

\begin{figure}[!h]
  \centering
      \includegraphics{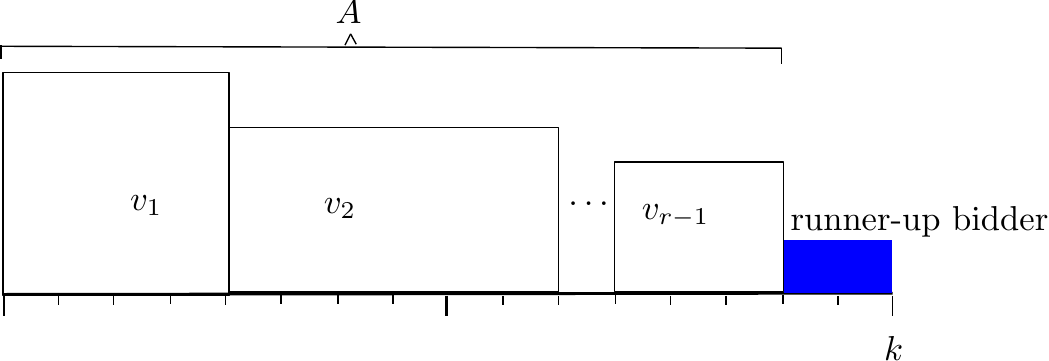}
  \caption{Each rectangle corresponds to a bidder where the height, width, and area represent \ppi{}, demand, and valuation of the bidder respectively. The dark rectangle corresponds to the runner-up bidder whose demand crosses the value $k$.}
  \label{fig:runner-up}
\end{figure}

\begin{observation}
We have $\ck \leq A < k$.
\end{observation}
\begin{proof}
Inequality $A < k$ is the direct result of the way we select runner-up bidder $r$. Inequality $\lceil \frac{k}{2} \rceil \leq A$ follows from the fact that $\sum_i^r d_i \geq k$ and the demand of the runner-up bidder is less than or equal to $\lfloor \frac{k}{2} \rfloor$ by definition of low-demand bidders.
\qed
\end{proof}
Now we are ready to precisely define how $\M$ selects and charges the set of winners. With probability $1/3$ $\M$ selects the most valuable high-demand bidder (which is the high-demand bidder with largest valuation breaking the ties by index). Therefore, the winning bidder in this case is $\ell + 1$ and she pays $v_{\ell + 2}$ which is the second highest valuation among high-demand bidders. Therefore her expected payment is $p^{\M}_{\ell + 1}(\theta) = \frac{v_{\ell + 2}}{3}$.

If we did not select the largest high-demand bidder then mechanism $\M$ uniformly at random selects the winner set from the first $r-1$ low-demand bidders where the probability of selecting each bidder $i \in [r-1]$ is $\frac{\ck}{A}$. In this case if bidder $i$ gets selected she has to pay $d_i \cdot \ppi_r$. Therefore, her expected payment is $p^{\M}_{i}(\theta) = \frac{2\ck}{3A} \cdot d_i \cdot \ppi_r$ since with probability $2/3$ we select low-demand-bidders and with probability  $\frac{\ck}{A}$ bidder $i$ gets selected. The probability $\frac{\ck}{A}$ is selected in a way such that if the low-demand bidders win, the expected number of allocated items is $\ck$ since the sum of demands of the first $r-1$ low-demand bidders is $A$ and each of them gets selected with probability $\frac{\ck}{A}$. 

In summary the expected payments of the bidders in mechanism $\M$ is the following. 
\begin{equation}
p^{\M}_i(\theta) = 
\begin{cases}
0 & \ell + 1 <i\\
\frac{v_{\ell + 2}}{3} & i = \ell + 1 \\
0 & r \leq i \leq \ell\\
\frac{2\ck}{3A} \cdot d_i \cdot \ppi_r & 1 \leq i < r \\
\end{cases}
\end{equation}

In the following first we prove that the allocation function of $\M$ is monotone and then we show that the unique expected payments of the winners calculated using \cref{thm:tf} is equal to the expected payment of mechanism $\M$ which proves that $\M$ truthful.
\begin{observation}
$w_i\tp{d_i, v_i}{\theta_{-i}}$ is monotone in $v_i$.
\end{observation}
\begin{proof}
If bidder $i$ is a high-demand bidder then clearly increasing her bid just increases her chance to be the high-demand bidder with the largest valuation and hence win with probability $1/3$. If bidder $i$ is a low-demand bidder then increasing her bid just increases her \ppi{} and hence can help her to go over the \ppi{} of the runner-up bidder and win with probability $\frac{2\lceil \frac{k}{2} \rceil}{3 \cdot A}$.
\qed
\end{proof}

The following lemma shows the expected payment of each winner.

\begin{lemma}
\label{lem:expected-payments}
The truthful expected payment of bidder $i$ ($p_i(\theta)$) calculated by \cref{cond:2} of \cref{thm:tf} is the following.
\begin{equation}
p_i(\theta) = 
\begin{cases}
0 & \ell + 1 <i\\
\frac{v_{\ell + 2}}{3} & i = \ell + 1 \\
0 & r \leq i \leq \ell\\
\frac{2\ck}{3A} \cdot d_i \cdot \ppi_r & 1 \leq i < r \\
\end{cases}
\end{equation}
\end{lemma}

\begin{proof}
Remember that the first $\ell$ bidders are low-demand bidders which have non-decreasing \ppi{}s, $r$ is the low-demand runner-up bidder, and finally among high-demand bidders, bidder $\ell + 1$ has the largest valuation and bidder $\ell + 2$ has the second largest valuation.

The probability of winning for bidder $i$ when $\ell + 1 <i$ is zero since she is a high-demand bidder who either does not have the highest valuation or has the highest valuation but has larger index number (see Definition \ref{def:valuable}). Because function $w_i$ is monotone we conclude that $w_i\tp{d_i, t}{\theta_{-i}}$ is equal to zero for any $t \leq v_i$. Hence by calculating the formula in \cref{cond:2} of \cref{thm:tf} we get $p_i(\theta) = 0$.

We calculate the truthful expected payment of bidder $\ell + 1$ by using the formula in \cref{cond:2} of \cref{thm:tf}.
\begin{align*}
\allowdisplaybreaks
p_{\ell + 1}(\theta) &= v_{\ell + 1} \cdot w_{\ell + 1}(\theta) - \int_{0}^{v_{\ell + 1}} w_{\ell + 1}\tp{d_{\ell + 1}, t}{\theta_{-{\ell + 1}}} dt \\
& = \frac{1}{3} v_{\ell + 1} - \int_{0}^{v_{\ell + 1}} w_{\ell + 1}\tp{d_{\ell + 1}, t}{\theta_{-{\ell + 1}}} dt \\
\displaybreak[2]
& = \frac{1}{3} v_{\ell + 1} - \int_{0}^{v_{\ell + 2}} w_{\ell + 1}\tp{d_{\ell + 1}, t}{\theta_{-{\ell + 1}}} dt - \int_{v_{\ell + 2}}^{v_{\ell + 1}} w_{\ell + 1}\tp{d_{\ell + 1}, t}{\theta_{-{\ell + 1}}} dt\\
\displaybreak[2]
& = \frac{1}{3} v_{\ell + 1} - \int_{v_{\ell + 2}}^{v_{\ell + 1}} w_{\ell + 1}\tp{d_{\ell + 1}, t}{\theta_{-{\ell + 1}}} dt\\
& = \frac{1}{3} v_{\ell + 1} - \frac{1}{3} (v_{\ell + 1} - v_{\ell + 2})\\
& = \frac{v_{\ell + 2}}{3}\\
\end{align*}
The first equality is \cref{cond:2} of \cref{thm:tf}, the second equality follows from the fact that probability of winning for bidder $\ell + 1$ ($w_{\ell + 1}(\theta)$) is $1/3$, the third one is breaking the domain of integration, the forth and fifth equalities are followed by noting that probability of winning for bidder $\ell + 1$ is zero if his valuation is less than $v_{\ell + 2}$ and is $1/3$ if his valuation is greater than or equal to $v_{\ell + 2}$.

The probability of winning for bidder $i$ when $r \leq i \leq \ell$ is zero since she is a low-demand bidder which has \ppi{} less than or equal to $\ppi_r$. Because function $w_i$ is monotone we conclude that $w_i\tp{d_i, t}{\theta_{-i}}$ is equal to zero for any $t \leq v_i$. Hence by calculating the formula in \cref{cond:2} of \cref{thm:tf} we get $p_i(\theta) = 0$.

The only part remaining is to show that $p_i(\theta) = \frac{2\ck}{3A} \cdot d_i \cdot \ppi_r$ for $1 \leq i < r$  using \cref{cond:2} of \cref{thm:tf}. In order to calculate $\int_{0}^{v_i} w_i\tp{d_i, t}{\theta_{-i}} dt$ we consider the curve of allocation function $w_i\tp{d_i, t}{\theta_{-i}}$ when $t$ increases from zero to $v_i$ (see \cref{fig:alloc-curve}). 
\begin{figure}[!h]
  \centering
      \includegraphics{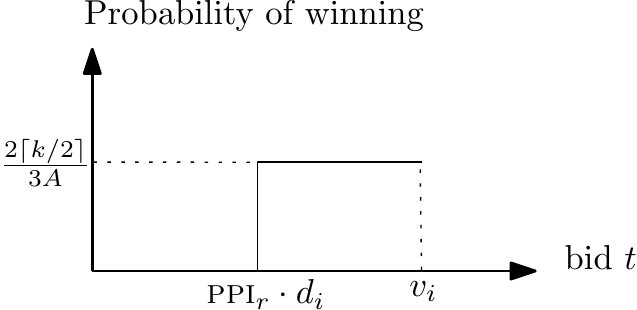}
  \caption{The horizontal axis represents the bid of bidder $i$ and the vertical axis shows the probability of bidder $i$ winning. As bidder $i$ increases her bid; at the point when her \ppi{} is equal to the \ppi{} of the runner-up bidder she gets allocated with probability $\frac{2\ck}{3A}$.}
  \label{fig:alloc-curve}
\end{figure}
\begin{observation}
\label{clm:ldb-curve}
For any bidder $1 \leq i < r$ allocation function $w_i\tp{d_i, t}{\theta_{-i}}$ is equal to zero when $t < d_i \cdot \ppi_r$ and is equal to $\frac{2\ck}{3A}$ when $t \geq d_i \cdot \ppi_r$.
\end{observation}
\begin{proof}
Remember the runner-up bidder $r$ has the smallest index for which $\sum_{j = 1}^r d_j \geq k$. Mechanism $\M$ allocates all the low-demand bidders which have index less than $r$ (have \ppi{}s greater than or equal to the runner-up bidder) with probability $\frac{2\ck}{3A}$. Therefore as far as $t \geq d_i \cdot \ppi_r$ bidder $i$ is more valuable than the runner-up bidder (see Definition \ref{def:valuable}) and wins with probability $\frac{2\ck}{3A}$ in type profile $\tp{d_i, t}{\theta_{-i}}$.

Now assume that $t < d_i \cdot \ppi_r$ and $\theta' = \tp{d_i, t}{\theta_{-i}}$. Our objective is to show that the probability of bidder $i$ winning is zero for type profile $\theta'$ and hence finish the proof of the observation. Note that $\sum_{j = 1}^r d_j \geq k$ and in the new type profile $\theta'$ bidder $i$ has \ppi{} less than the $\ppi$s of all bidders $j \in [r]$ where $j \neq i$ since $t < d_i \cdot \ppi_r$. In other words, bidder $i$ is the least valuable bidder in $[r]$ (see Definition \ref{def:valuable}) while $\sum_{j = 1}^r d_j \geq k$. Therefore, bidder $i$ is either the runner-up bidder in $\theta'$ or has \ppi{} less than the runner-up bidder (see  \cref{def:runner-up}). Hence has zero probability of winning.
\qed
\end{proof}
The following equalities  shows the expected payment of bidder $i$ for $1 \leq i < r$.
\begin{align*}
\allowdisplaybreaks
p_i(\theta) &= v_i \cdot w_i(\theta) - \int_{0}^{v_i} w_i\tp{d_i, t}{\theta_{-i}} dt \\
& = \frac{2\ck}{3A} v_i - \int_{0}^{v_i} w_i\tp{d_i, t}{\theta_{-i}} dt \\
& = \frac{2\ck}{3A} v_i - \int_{0}^{d_i \cdot \ppi_r} w_i\tp{d_i, t}{\theta_{-i}} dt - \int_{d_i \cdot \ppi_r}^{v_i} w_i\tp{d_i, t}{\theta_{-i}} dt\\
& = \frac{2\ck}{3A} v_i - \int_{d_i \cdot \ppi_r}^{v_i} w_i\tp{d_i, t}{\theta_{-i}} dt\displaybreak[2]\\
& = \frac{2\ck}{3A} v_i - \frac{2\ck}{3A} (v_i - d_i \cdot \ppi_r)\\
& = \frac{2\ck}{3A} (d_i \cdot \ppi_r)\\
\end{align*}
The first equality is \cref{cond:2} of \cref{thm:tf}, the second equality follows from the fact that probability of winning for bidder $i$ ($w_i(\theta)$) is $\frac{2\ck}{3A}$, the third one is breaking the domain of integration, the forth and fifth equalities are followed from \cref{clm:ldb-curve}.
\qed
\end{proof}

Let $\rev(\M, \theta)$ denotes the expected revenue of mechanism $\M$ in type profile $\theta$. We prove the following.

\begin{align}
\rev(\M, \theta) &= \sum_{i = 1}^n p_i(\theta) & \text{definiton of \rev{}} \notag\\
&= \sum_{i = 1}^{r-1} p_i(\theta) + \sum_{i = r}^{\ell} p_i(\theta) + p_{\ell + 1}(\theta) + \sum_{i = \ell + 2}^{n} p_i(\theta)&\notag\\
&= \sum_{i = 1}^{r-1} \frac{2\ck}{3A} (d_i \cdot \ppi_r) + \frac{v_{\ell + 2}}{3} & \text{Lemma \ref{lem:expected-payments}}\notag\\
&= \frac{2\ck}{3A}\cdot \ppi_r \cdot \sum_{i = 1}^{r-1} d_i + \frac{v_{\ell + 2}}{3} & \notag\\
&= \frac{2\ck}{3}\cdot \ppi_r + \frac{v_{\ell + 2}}{3} & \text{as $A = \sum_{i = 1}^{r-1} d_i$}\label{eqn:expected-rev}
\end{align}

The following lemma proves that $\M$ is revenue monotone.
\begin{lemma}
The expected revenue of mechanism $\M$ does not decrease if we add one more bidder or a bidder increases her bid.
\end{lemma}
\begin{proof}
The expected revenue of mechanism $\M$ is $\frac{2\ck}{3}\cdot \ppi_r + \frac{v_{\ell + 2}}{3}$ by \cref{eqn:expected-rev}. Remember that $v_{\ell + 2}$ is the high-demand bidder with the second highest valuation, therefore, if we add one more bidder or a bidder increases her bid this value does not decrease. On the other hand $\ppi_r$ is the \ppi{} of the runner-up bidder (see \cref{def:runner-up}) which is the \ppi{} of the first low-demand bidder that crosses value $k$ (see \cref{fig:runner-up}) where the bidders are sorted according to their \ppi{}s. The proof of the lemma follows by the fact that $\ppi_r$ also does not decrease as we add one more bidder or a bidder increases her bid.
\qed
\end{proof}

In the following lemma we prove that Price of Revenue Monotoncity (\porm) of $\M$ is $3$ and  finish this section.
\begin{lemma}
\label{lm:porm}
$\porm(\M) = 3$.
\end{lemma}

\begin{proof} 
We prove the lemma by showing that the expected social welfare of $\M$ in type profile $\theta$ is at least $\frac{1}{3}$ of the maximum social welfare ($\wf(\theta)$). Let $\Sy$ be an arbitrary subset of bidders which VCG selects and realizes the maximum social welfare $\wf(\theta)$. Let $L$ be the sum of valuations of low-demand bidders in $\Sy$ and $H$ be the sum of valuations of high-demand bidders in $\Sy$. Therefore, we have:
\begin{equation}
\label{eq:porm-1}
\wf(\theta) = L + H
\end{equation}
There can be at most one high-demand bidders in $\Sy$ since they have demand more than $\fk$. As $v_{\ell + 1}$ is the high-demand bidder with the largest valuation we have the following.
\begin{equation}
\label{eq:porm-2}
v_{\ell + 1} \geq H
\end{equation}
We also have 
\begin{equation}
\label{eq:porm-3}
\sum_{ i = 1}^{r-1}v_i \geq L \cdot \frac{A}{k}
\end{equation}
since $A = \sum_{i = 1}^{r-1} d_i$ and the first $r-1$ bidders have larger \ppi{}s than the rest as they are sorted non-increasingly according to their \ppi{}s.

Remember that $\M$ selects bidder $\ell + 1$ with probability $1/3$ or selects each of the first $r-1$ low-demand bidders with probability $\frac{2\ck}{3A}$.
The following equalities finishes the proof of the lemma.
\begin{align*}
E[\wf(\M, \theta)] &= \sum_{i = 1}^{r-1} \frac{2\ck}{3A} \cdot v_i + \frac{1}{3} v_{\ell + 1} & \text{definition of expected social welfare} \\
&= \frac{2\ck}{3A} \cdot \sum_{i = 1}^{r-1} v_i + \frac{1}{3} v_{\ell + 1} \\
&\geq \frac{2\ck}{3A} \cdot L \cdot \frac{A}{K} + \frac{1}{3} H & \text{by \cref{eq:porm-2} and \cref{eq:porm-3}}\\
&\geq \frac{1}{3} L + \frac{1}{3} H & \text{algebra}\\
&= \frac{1}{3} \wf(\theta) & \text{by \cref{eq:porm-1}}\\
\end{align*}
\end{proof}

%% file: multigroup.tex
\newcommand{\tpg}[2]{\left(d^{(#1)}_{#2},v^{(#1)}_{#2}\right)}
\newcommand{\ppig}[2]{\ppi^{(#1)}_{#2}}
\newcommand{\iv}[2]{u^{(#1)}_{#2}}
\newcommand{\ival}[1]{$#1$th-item value}
\newcommand{\vv}[2]{v^{(#1)}_{#2}}
\newcommand{\dd}[2]{d^{(#1)}_{#2}}
\newcommand{\mmca}[1]{\mbox{\sc mmca}}
\section{Multigroup Combinatorial Auction with Identical Items}
\label{sec:mcaii}
In this section we describe our Mechanism for Multigroup Combinatorial Auction with identical items (\mmca{}). We prove that \mmca{} satisfies IC and RM while has \porm{} of at most $O(\log k)$.

In the following we define required notations to be used throughout this section.  Let's assume we have $m$ groups $G^{(1)}, \ldots, G^{(m)}$. We always show the group index of any variable within parenthesis in superscript. For group $G^{(g)}$ let $\left(D^{(g)},V^{(g)}\right)$ be the type of a high-demand bidder with the highest valuation, $\tpg{g}{1}, \tpg{g}{2}, \ldots, \tpg{g}{n^{(g)}}$ be the types of low-demand bidders, and $\ppig{g}{i}$ be the \ppi{} of $i$th low-demand bidder. Without loss of generality, for each group $G^{(g)}$ we assume that the number of low-demand bidders $n^{(g)}$ is larger than $k$ and $\ppig{g}{1} \geq \ppig{g}{2} \geq \ldots\geq \ppig{g}{n^{(g)}}$.  
The following defines the maximum price per item with which we can sell at least $j$ items to the low-demand bidders of group $G ^{(g)}$.
\begin{definition}[\ival{j}]
\label{def:ival}
For each $j \in [k]$ we define {\em \ival{j}} ($\iv{g}{j}$) of group $G^{(g)}$ to be equal to the valuation of $i$th low-demand bidder $\vv{g}{i}$ where $i$ is the minimum number for which $\sum_{t = 1}^i \dd{g}{t}$ is greater than or equal to $j$ (see \cref{fig:itemvalue}).
\end{definition}

\begin{figure}[!h]
  \centering
      \includegraphics{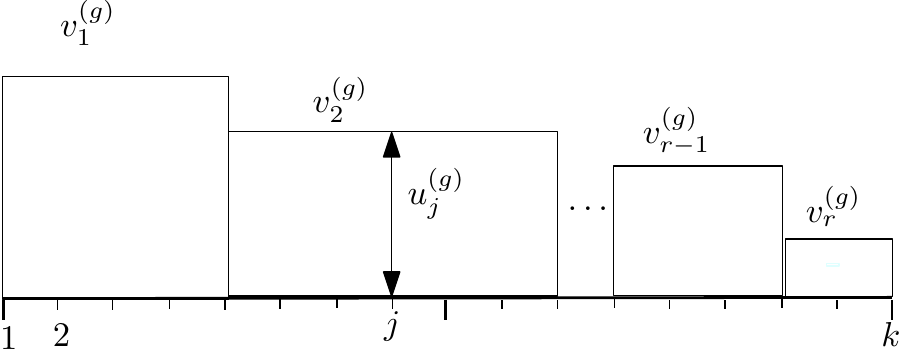}
  \caption{Each rectangle corresponds to a low-demand bidder of group $G^{(g)}$ where the height, width, and area represent \ppi{}, demand, and valuation of the bidder respectively. Here bidders are sorted according to their \ppi{}s and the valuation of the bidder who crosses item $j$ is the $j$th item-value.}
  \label{fig:itemvalue}
\end{figure}

We are interested in assigning a value to each group which represents how much revenue can be obtained if we give the items to bidders of the group. Then, we give the items to bidders of a group with the highest assigned value.

\begin{defi}[$\mprg^{(g)}$] The Maximum Possible Revenue of Group $G^{(g)}$ ($\mprg^{(g)}$) is equal to 
\[
\max(V^{(g)}, \max_{j \in [\ck]} j \cdot \iv{g}{j}).
\]
\end{defi} 

The intuition for $\mprg^{(g)}$ is the following. The value $j \cdot \iv{g}{j}$ is the maximum revenue we can obtain if we sell exactly $j$ items to low-demand bidders of group $G^{(g)}$, see Definition \ref{def:ival}. Note that number $j$ is taken from the set $[\ck]$ meaning that we consider selling at most $\ck$ items. This is because low-demand bidders have different demands from the range $[1..\ck]$, therefore, we can guarantee selling at most $\ck$ items without overselling the $k$ items. In fact with randomization we make sure that we sell exactly $\ck$ items in expectation. Finally we take the maximum of the highest valuation of high-demand bidders ($V^{(g)}$) and value $\max_{j \in [\ck]} j \cdot \iv{g}{j}$.

The rest of this section is organized as follow.
First, we describe the allocation function of \mmca{}. Second, we use Lemma \ref{lem:expected-payments} to derive the expected payments of the winners which determines the revenue of \mmca{}. Third, we show that the revenue does not decrease if we add one more bidder or a bidder increases her bid. Finally, we show that \porm{} of \mmca{} is at most $O(\log k)$.

Let $G^{g^*}$ be the group with highest $\mprg{}$, $g^* = \arg \max_g \mprg^{(g)}$, and $G^{\hat{g}}$ be the group with the second highest $\mprg{}$. Mechanism \mmca{} selects the winners from group $G^{(g^*)}$. Let $R$ be equal to $\mprg^{\hat{g}}$. We think of $R$ as a reserved value such that we must obtain at least $R$ revenue from group $G^{(g^*)}$. In the rest of this section all the discussions are about group $G^{(g^*)}$ unless mentioned otherwise, henceforth, we drop the group identifiers from variables.

Similar to Section \ref{sec:prelim}, in the following, we define {\em runner-up} bidder to be the low-demand bidder with highest \ppi{} which cannot be a winner if we sort bidders by their \ppi{}s (see \cref{fig:runner-up-g}).

\begin{figure}[!h]
  \centering
      \includegraphics{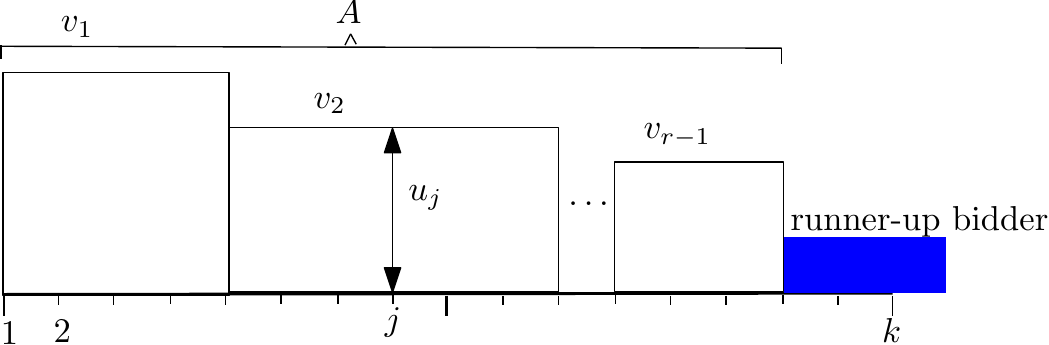}
  \caption{Each rectangle corresponds to a low-demand bidder of group $G^{(g^*)}$ where the height, width, and area represent \ppi{}, demand, and valuation of the bidder respectively. Here bidders are sorted according to their \ppi{}s and the low-demand bidder whose demand crosses $k$th item is the runner-up bidder.}
  \label{fig:runner-up-g}
\end{figure}

\begin{definition}[runner-up bidder]
\label{def:runnerg}
We call low-demand bidder $r$ the runner-up bidder if $r$ is the smallest number in set $[n]$ for which 
\[\sum_{i = 1}^{r} d_{i} \geq k.\]
\end{definition}

Remember that $R$ is equal to the second largest \mprg{} of groups, $V$ is the largest valuation of high-demand bidders of the winning group $G^{(g^*)}$, and $u_j$ is largest price of an item for which we can sell at least $j$ items to low-demand bidders of $G^{(g^*)}$ (see Definition \ref{def:ival}).  Note that as $\mprg = \max(V, \max_{j \in [\ck]} j \cdot u_j)$ of $G^{(g^*)}$ is larger than $R$ we have either: (1) $V > \max(R, \max_{j \in [\ck]} j \cdot u_j)$ or (2) $\max_{j \in [\ck]} j \cdot u_j \geq \max(R, V)$. 

Case (1) is easy, if $V > \max(R, \max_{j \in [\ck]} j \cdot u_j)$ then the set of winners contains only a high-demand bidder with valuation $V$ breaking the ties arbitrarily. In Case(2) let $j^*$ be the largest number from $[\ck]$ such that $j^* \cdot u_{j^*} \geq \max(R, V)$. If $j^*$ is less than $\ck$ then the set of winners is all low-demand bidders whose \ppi{} is greater than or equal to $u_{j^*}$. Otherwise if $j^*$ is equal to $\ck$ then we need to include all the low-demand bidders whose \ppi{} is  larger than or equal to $\max(R, V)/\ck$ since, roughly speaking, their \ppi{} is high enough to win against both high-demand bidders and the group with second highest \mprg{}. Note that the sum of demands of such bidders might exceed $\ck$, therefore, we need randomization to guarantee selling exactly $\ck$ items in expectation.

\begin{definition}[$a$ and $A$]
\label{def:ag}
If $ \ck \cdot u_{\ck} \geq \max(R, V)$, we define $a$ to be the largest number in $[r-1]$ such that $\ppi_a \geq \max(R, V)/\ck$, \ie, number $a$ is the smallest index in the set of all low-demand bidders that have index greater than the runner-up bidder and have $\ppi$ larger than or equal to $u_{\ck}$. We also define number $A$ to be the sum of demands of the first $a$ low-demand bidders, \ie, $A = \sum_{i = 1}^{p} \dd{g}{i}$. 
\end{definition}

Now we are ready to formally define the allocation function of \mmca{}.

\begin{definition}[Allocation Function of \mmca{}]
\label{def:mmcaalloc}
\begin{enumerate}
\item \label{cs:hd} If $V > \max(R, \max_{j \in [\ck]} j \cdot u_j)$ then the set of winners contains only a high-demand bidder with valuation $V$ breaking the ties arbitrarily.
\item \label{cs:lds} If $j^*$ is less than $\ck$ then the set of winners is all low-demand bidders whose \ppi{} is greater than or equal to $u_{j^*}$.
\item \label{cs:ldl} If $j^*$ is equal to $\ck$ then each of the first $a$ low-demand bidder wins with probability $\frac{\ck}{A}$ independently.
\end{enumerate}
\end{definition}

In the following lemma we calculate the critical values of winners using \cref{thm:tf}.

\begin{lemma}
\label{lm:mmcacv}
The critical values of winners are the following considering different conditions of Definition \ref{def:mmcaalloc}.
\begin{enumerate}
\item If \cref{cs:hd} happens then the critical value of the winner is $\max(R, \max_{j \in [\ck]} j \cdot u_j, V_2)$ where $V_2$ is the second highest valuation of high-demand bidders.
\item If \cref{cs:lds} happens then the critical value of each winner $i$ is $d_i \cdot (\max(R, V)/j^*)$.
\item If \cref{cs:ldl} happens then the critical value of each winner $i$ is $\frac{\ck}{A} \cdot d_i \cdot  \max(\ppi_r, (\max(R, V)/\ck))$. 
\end{enumerate}
\end{lemma}
\begin{proof}
We consider all the three conditions separately.

{\bf \cref{cs:hd}:}  We show that if the valuation $V$ of the winner goes below $\max(R, \max_{j \in [\ck]} j \cdot u_j, V_2)$ then player $V$ cannot be a winner. If the value of $\max(R, \max_{j \in [\ck]} j \cdot u_j, V_2)$ is equal to $R$ then decreasing $V$ to a value less than $R$ causes a group with the second highest \mprg{} win and hence changes player $V$ to a loser participant. If it is equal to $\max_{j \in [\ck]} j \cdot u_j$ then decreasing $V$ to a value less than $\max_{j \in [\ck]} j \cdot u_j$ causes low-demand bidders win and hence changes player $V$ to a loser participant. If it is equal to $V_2$ then decreasing $V$ to a value less than $V_2$ causes the high-demand bidder with valuation $V_2$ win and hence changes player $V$ to a loser participant.

{\bf \cref{cs:lds}:} Remember that $j^*$ is the largest number from $[\ck]$ such that $j^* \cdot u_{j^*} \geq \max(R, V)$. Therefore if the valuation of any winner $i$ goes below $d_i \cdot (\max(R, V)/j^*)$ then there exist no $j$ in $[\ck]$ such that $j \cdot u_{j}$ is larger than $\max(R, V)$. Hence the winning set changes to either another group if $\max(R, V) = R$ or to high-demand bidders if $\max(R, V) = V$.

{\bf \cref{cs:ldl}:} In this case we show that if the valuation of winner $i$ goes below $d_i \cdot  \max(\ppi_r, (\max(R, V)/\ck))$ then she has zero probability of winning and if it is larger than or equal to $d_i \cdot  \max(\ppi_r, (\max(R, V)/\ck))$ then she has $\frac{\ck}{A}$ probability of winning. If the valuation of $i$ is more than $d_i \cdot  \max(\ppi_r, (\max(R, V)/\ck))$ then by the way we define the allocation function (see Definition \ref{def:mmcaalloc}) he has probability $\frac{\ck}{A}$ of winning otherwise her \ppi{} is less than $\max(\ppi_r, (\max(R, V)/\ck))$. If \ppi{} of $i$ is less than $\ppi_r$ (the \ppi{} of the runner-up bidder) then she cannot be a winner since the sum of demands of participants who have higher $\ppi$ than her is larger than $k$. If \ppi{} of $i$ is less than $\max(R, V)/\ck$ then she cannot win because of the way we select winners in the allocation function (\cref{cs:ldl} of Definition \ref{def:mmcaalloc}).
\qed
\end{proof}

Now we prove that \mmca{} satisfies RM.

\begin{lemma}
If we add one more bidder or a bidder increases her bid the revenue of \mmca{} does not decrease.
\end{lemma}
\begin{proof}
Let $x$ be the new participant or the participant who has increased her bid. Let $\theta$ be the type profile before adding $x$ and $\theta'$ be the type profile after adding $x$. We need to prove that $\rev(\mmca{}, \theta') \geq \rev(\mmca{}, \theta)$. First we prove that the revenue of \mmca{} is between the \mprg{} of the highest group and the second highest group.
\begin{observation}
\label{clm:between}
We have $R \leq \rev(\mmca{}, \theta) \leq \mprg^{g^*}$.
\end{observation}
\begin{proof}
If \cref{cs:hd} of Definition \ref{def:mmcaalloc} happens then the revenue of \mmca{} is $\max(R, \max_{j \in [\ck]} j \cdot u_j, V_2)$ by Lemma \ref{lm:mmcacv}. Note that in this case the revenue is less than $V$ and more than $R$ and hence the proof of the observation follows.

If \cref{cs:lds} of Definition \ref{def:mmcaalloc} happens then the revenue of \mmca{} is $\max(R, V)$ by Lemma \ref{lm:mmcacv} since we sell $j^*$ items. Note that in this case the revenue is less than $\max_{j \in [\ck]} j \cdot u_j$ and more than $R$ and hence the proof of the observation follows.

If \cref{cs:ldl} of Definition \ref{def:mmcaalloc} happens then the revenue of \mmca{} is $\max(\ck \cdot \ppi_r, R, V)$ by Lemma \ref{lm:mmcacv} since we sell $\ck$ items in expectation. Note that in this case the revenue is less than $\max_{j \in [\ck]} j \cdot u_j$ and more than $R$ and hence the proof of the observation follows.
\qed
\end{proof}
By \cref{clm:between} we know that the revenue of \mmca{} is between \mprg{} of the highest group and \mprg{} of the second highest group. Therefore if we add participant $x$ and the group with highest \mprg{} changes then it means that the revenue in $\theta'$ is now $\mprg^{(g^*)}$ and hence increases. 

Now we assume that after adding $x$ the winning group does not change. Here we can check all the three cases that can happen in allocation function of \mmca{} (Definition \ref{def:mmcaalloc}) for new type profile $\theta'$ and see that adding $x$ can only increase the revenue.
\qed
\end{proof}

Now we prove that \porm{} of \mmca{} is at most $O(\ln k)$.

\begin{lemma}
$\porm{}(\mmca{}) \leq (2 + \ln(k))$
\end{lemma}
\begin{proof}
First we show that for any group $g$ \mprg{} of $g$ is at least $\frac{1}{2 + \ln(k)}$ fraction of the maximum social welfare obtainable by group $g$.
\begin{observation}
\label{clm:mprgwf}
For any group $G^{(g)}$ we have $\mprg^{(g)} \geq \frac{1}{(2 + \ln(k))} \cdot \wf(G^{(g)})$.
\end{observation}
\begin{proof}
Let set $S$ contain bidders of $G^{(g)}$ which obtain the maximum social welfare and $T$ be equal to $\mprg^{(g)}$. Note that set $S$ can contain at most one high-demand bidder since their demand is larger than $\ck$. Moreover, the valuation of high-demand bidder of $S$ cannot be more than $T$ because otherwise \mprg{} of $G^{(g)}$ will be larger than $T$. Therefore the total social welfare of $S$ from high-demand bidders is at most $T$.

Now let us sort the low-demand bidders of $S$ by their \ppi{} and define value $u_{S,j}$ to be the maximum price per item with which we can sell at least $j$ items to the low-demand bidders of $S$, similar to Definition \ref{def:ival}. Note that for any $j \in [\ck]$ we have $u_{S,j} \leq \frac{T}{j}$ because otherwise \mprg{} of $G^{(g)}$ will be larger than $T$. Moreover for any $j \in \{\ck + 1 , \ldots, k\}$ we have $u_{S,j} \leq \frac{T}{\ck}$ because we sort by \ppi{}. Therefore the social welfare of $S$ from low-demand bidders is at most $\sum_{j=1}^{\ck} \frac{T}{j} + \ck \cdot \frac{T}{\ck}$ which is at most $(1 + \ln(k)) \cdot T$.

Set $S$ can get social welfare of at most $T$ from the high-demand bidders and $(1 + \ln(k)) \cdot T$ from low-demand bidders, therefore, the proof of the observation follows.
\qed
\end{proof}
Note that by the way we define the allocation function of \mmca{} it always selects a set of winners such that their expected social welfare is at least the \mprg{} of the winning group, see Definition \ref{def:mmcaalloc}. Since \mmca{} selects a group with the highest \mprg{} as the winning group by \cref{clm:mprgwf} we know that the maximum social welfare cannot be more than $(2 + \ln(k))$ times the highest \mprg{} and hence the proof of the lemma follows.
\qed
\end{proof}

Now we provide evidence that no randomized algorithm can obtain \porm{} better than $\Omega(\ln(k))$. Note that an optimum randomized mechanism ($\M^*$) has to first give a probability distribution to groups and pick the winning group according to that distribution. Let's assume $p^*$ be the function which maps each group $g$ to its probability of winning. Now we argue that $p^*$ cannot be dependent to the social welfare of $g$ otherwise $\M^*$ will not satisfy RM. Because, in this case a bidder can increase her bid high enough so that probability $p^*(g)$ goes to its upper-limit which decrease the critical value of the other bidders and hence break RM. We further argue that $p^*$ should be dependent to one factor of each group which is the same across the group. For example if it is dependent to two factors then increasing the first factor might remove the load from the second factor and hence the second factor is free to go down without changing the probability. If $p^*$ is dependent to only one factor then the best way to make it as close as possible to social welfare is to define $j \cdot u^{(g)}_j$ (see Definition \ref{def:ival}) for each group similar to our \mprg{} value assigned to each group. Therefore, we have evidence that function $p^*$ can be dependent to a value which can be off from social welfare by factor $\ln(k)$.

Now consider the simple case of image-text auction where we have two groups: text-ads and image-ads. We want to assign the winning probabilities to each group. Now suppose we are given two values $a$ for image-ads and $b$ for text-ads, and also we know that social welfare of text-ads is either $b$ or $\ln(k)\cdot b$. If we pick each group with probability $1/2$ and $1/2$, the expected social welfare is $1/2$ approximation to its max value. The question is can we do better by having a more clever randomization.

We prove that no other randomization can give us a factor better than $1/2 + 1/(2 \sqrt{\ln (k)})$.  Suppose the value of $a$ is equal to $X$ and the value of $b$ is equal to $X/\sqrt{\ln k}$ where the social welfare of text-ads can be $X \cdot \sqrt{\ln (k)}$. Now suppose $\M^*$ gives the items to image-ads with probability $F$ and to text-ads with probability $1-F$. If the social welfare  of text-ads is $X/\sqrt{\ln (k)}$, we get \porm{} of $F + (1-F)/\sqrt{\ln (k)}$ and if the social welfare of text-ads is $X \cdot \sqrt{\ln (k)}$, we get \porm{} of $(1-F) + F/\sqrt{\ln (k)}$. If we want to maximize the minimum of the two \porm{}s we have to set $F$ to $1/2$ which gives \porm{} of $1/2 + 1/(2 \sqrt{\ln (k)})$. Note that if the number of groups increases to $m$ then this factor changes to $1/m + 1/(m \sqrt{\ln (k)})$, therefore, the best way is to give probability one to the group with the best assigned value and lose factor $\ln(k)$.